\documentclass[a4paper,onecolumn,11pt,accepted=2023-01-10]{quantumarticle}
\pdfoutput=1

\usepackage[dvips]{graphicx}
\usepackage{amsmath,amssymb,amsthm,mathrsfs,amsfonts,dsfont}
\usepackage{multirow}
\usepackage{dcolumn}
\usepackage{bm}
\usepackage{amsmath}
\usepackage{amssymb}
\usepackage{braket}
\usepackage{physics}
\usepackage{amsthm}
\usepackage{mathtools}
\usepackage{diagbox}
\usepackage[utf8]{inputenc}
\usepackage[english]{babel}
\usepackage{stackengine}
\usepackage[noend,noline,ruled,linesnumbered]{algorithm2e}
\usepackage[numbers,sort&compress]{natbib}
\usepackage[colorlinks,citecolor=blue,linkcolor=blue,bookmarks=true]{hyperref}

\newtheorem{proposition}{Proposition}

\newtheorem{definition}{Definition}

\newcommand{\vabs}[1]{\left\| #1 \right\|}
\newcommand{\pbra}[1]{\left( #1 \right)}
\newcommand{\cbra}[1]{\left\{ #1 \right\}}
\newcommand{\sbra}[1]{\left[ #1 \right]}

\newcommand{\floor}[1]{\left\lfloor #1\right \rfloor }

\newcommand{\Ebb}{\mathbb{E}}

\newcommand{\Ibb}{\mathbb{I}}

\newcommand{\Kcal}{\mathcal{K}}
\newcommand{\Mcal}{\mathcal{M}}
\newcommand{\Scal}{\mathcal{S}}

\newcommand{\Pcal}{\mathcal{P}}

\newcommand{\Omat}{\bm{\mathrm O}}

\newcommand{\Pjmat}{P^{(j)}}

\newcommand{\Qjmat}{Q^{(j)}}

\newcommand{\supp}{\mathrm{supp}}

\newcommand{\Var}[1]{\text{Var} \left( #1 \right)}

\def\be{\begin{eqnarray}}
\def\ee{\end{eqnarray}}

\begin{document}
\title{Overlapped grouping measurement: A unified framework for measuring quantum states}
\date{January 11th 2023}
\author{Bujiao Wu}
\thanks{These two authors contributed equally }
\affiliation{Center on Frontiers of Computing Studies, Peking University, Beijing 100871, China}
\affiliation{School of Computer Science, Peking University, Beijing 100871, China}

\author{Jinzhao Sun}
\thanks{These two authors contributed equally }
\affiliation{Clarendon Laboratory, University of Oxford, Parks Road, Oxford OX1 3PU, United Kingdom}
\affiliation{Center on Frontiers of Computing Studies, Peking University, Beijing 100871, China}

\author{Qi Huang}
\affiliation{School of Physics, Peking University, Beijing 100871, China}
\affiliation{Center on Frontiers of Computing Studies, Peking University, Beijing 100871, China}

\author{Xiao Yuan}
\email{xiaoyuan@pku.edu.cn}
\affiliation{Center on Frontiers of Computing Studies, Peking University, Beijing 100871, China}
\affiliation{School of Computer Science, Peking University, Beijing 100871, China}

\begin{abstract}
    Quantum algorithms designed for 
    realistic quantum many-body systems, such as chemistry and materials,
    usually require a large number of measurements of the Hamiltonian. Exploiting different ideas, such as {importance sampling,} observable compatibility, or classical shadows of quantum states, different advanced measurement schemes have been proposed to greatly reduce the large measurement cost. Yet, the underline cost reduction mechanisms seem distinct from each other, and how to systematically find the optimal scheme remains a critical challenge.  Here, we address this challenge by proposing a unified framework of quantum measurements, incorporating advanced measurement methods as special cases. Our framework allows us to introduce a general scheme~---~overlapped grouping measurement, which simultaneously exploits the advantages of most existing methods. An intuitive understanding of the scheme is to partition the measurements into overlapped groups with each one consisting of compatible measurements. We provide explicit grouping strategies and numerically verify its performance for different molecular Hamiltonians with up to 16 qubits. Our numerical result shows significant improvements over existing schemes. Our work paves the way for efficient quantum measurement and fast quantum processing with current and near-term quantum devices. 
\end{abstract}

\maketitle

\section{Introduction}
\label{sec:intro}
How to efficiently measure a quantum state is a fundamental problem with great practical relevance.
Algorithms tailored for noisy intermediate-scale quantum (NISQ) devices~\cite{preskill2018quantum,bharti2021noisy,endoreview,cerezo2020variational} usually require measuring complicated multi-qubit observables, such as the Hamiltonian~\cite{peruzzo2014variational,o2016scalable,mcclean2017hybrid,li2017efficient,moll2018quantum,Liu2019Variational,wang2020noise,ma2020quantum,mcardle2019variational,Yang17Optimizing,Endo20Variational,Higgott2019variational,Zhou20Quantum,kandala2017hardware,grimsley2019adaptive,PhysRevX.8.011021,hempel2018quantum,XU2021,yuan2020quantum,yuan2019theory,fujii2020deep,nakanishi2019subspace,cirstoiu2019variational,gibbs2021long,commeau2020variational,bravo-prieto2019,huang2019near,mcardle2018quantum,Cao2019Quantum,sun2021perturbative}. 
When decomposing the Hamiltonian into local measurable observables, it may contain a large number of terms. For example, an electronic Hamiltonian with Coulomb interaction generally has $\mathcal O(M^4)$ terms when represented with $M$ fermionic modes~\cite{mcardle2018quantum,Cao2019Quantum,babbush2018low,arute2020hartree}.
The number of terms could already become quite large when we consider the classical limit with a large $M$, where the naive strategy of equally measuring all the observables requires a prohibitively long time. 
We also need error mitigation techniques to suppress calculation errors, which again introduces a large sample overhead with increasing problem size~\cite{li2017efficient,endo2018practical,temme2017error,mcclean2017hybrid,strikis2020learning,bravyi2020mitigating,sun2020mitigating2,dumitrescu2018cloud,otten2019accounting,mcclean2020decoding,giurgica2020digital}.
Therefore, an efficient quantum measurement scheme is crucial for demonstrating a clear and robust quantum advantage with NISQ devices.

Without introducing additional entangling circuits, 
three types of advanced  measurement schemes have been proposed to reduce the measurement cost by exploiting different features of the to-be-measured observables~\cite{kandala2017hardware,verteletskyi2020measurement,huang2020predicting,wecker2015progress,hadfield2020measurements,huang2021efficient,torlai2018neural,choo2020fermionic,torlai2020precise,cotler2020quantum,crawford2021efficient}.
First, observables may have different weight coefficients and we can exploit importance sampling to distribute more measurements to observables with large weights~\cite{wecker2015progress,mcclean2016theory}.
Next, observables may be compatible with some other ones, in the sense that they could be simultaneously measured with the same measurement basis. We can thus group observables into sets of compatible observables using fewer measurements~\cite{kandala2017hardware,verteletskyi2020measurement,vallury2020quantum,izmaylov2019unitary,zhao2020measurement,hempel2018quantum,o2016scalable,crawford2021efficient}.
Another notable but conceptually different scheme considers classical shadows of quantum states using uniformly random local measurements, which are extensively investigated in theoretical and experimental works~\cite{huang2020predicting,aaronson2019shadow,struchalin2021experimental,chen2020robust,zhao2020fermionic,acharya2021informationally,zhang2021experimental}. By properly post-processing the classical measurement outcomes, one can simultaneously obtain the expectation values of any observables. The cost of the original uniform classical shadow scheme~\cite{huang2020predicting} scales exponentially to the number of qubits that the observable non-trivially acts on, and later LBCS and {derandomized CS} methods were further proposed to reduce the measurement cost~\cite{hadfield2020measurements,huang2021efficient}. 
While the optimized classical shadow method outperforms the other two types of methods in numerical experiments, how they are related, and how to find an optimized method that exploits the advantages of all these advanced measurement schemes remain open.

Here, we address these problems in quantum state measurement. We first introduce a unified framework that integrates the advantages of the typically advanced measurement schemes in Sec. \ref{sec:unified_frame}. In particular, we show how to understand the classical shadow method as a generalized observable grouping method. We next introduce the overlapped grouping measurement scheme that simultaneously exploits the features of the importance sampling, observable compatibility, and classical shadows in Sec. \ref{sec:overlapped_group}. While finding the optimal overlapped groups could be computationally challenging, we provide explicit algorithms that output an optimized measurement scheme in Sec.~\ref{sec:explicit_strategy}.  We then numerically benchmark our method in Sec.~\ref{sec:numerical_test} by comparing it to existing advanced works~\cite{verteletskyi2020measurement,huang2020predicting,hadfield2020measurements,huang2021efficient,kandala2017hardware} in estimating expectations of molecular Hamiltonians, as a subroutine in most quantum algorithms. Our numerical result shows prominent improvements over all the others. The proposed method is immediately applicable to currently available and near-future quantum computing experiments. In Sec. \ref{sec:discussion}, we conclude this work and suggest some interesting future investigations.

\section{A unified framework}
\label{sec:unified_frame}
{Now we introduce a framework for measuring hermitian objective observables $\Omat:=\sum_{j} \alpha_j Q^{(j)}$ on a multi-qubit quantum state $\rho$. Here $Q^{(j)}\in \cbra{I,X,Y,Z}^n$ are tensor products of single-qubit Pauli operators, and we also call $Q^{(j)}$ local Pauli strings. }
Naively, we could measure each term $Q^{(j)}$ to obtain the expectation value, $\tr(\rho Q^{(j)})$, and hence the expectation value of the objective observable, $\tr(\rho \Omat)$,
whereas more efficient schemes may be found by exploiting the properties of the objective observables.
 
We first consider observable compatibility. 
Let $Q$ and $R$ be tensor products of single-qubit Pauli operators $Q=\bigotimes_{i=1}^n Q_i$ and $R = \bigotimes_{i = 1}^n R_i$ with $Q_i, R_i\in \cbra{I,X,Y,Z}$. We let $Q\triangleright R$ denote $Q_i = R_i$ or $Q_i = I$ for any $i${, indicating that measuring observable $R$ equivalently measures $Q$.}
We say that $Q$ is \emph{compatible} with $R$ when $Q_i\triangleright R_i$ or $R_i\triangleright Q_i$ for any $i$, meaning that $Q$ and $R$ can be simultaneously measured.
In an extreme case, when each operator $Q^{(j)}$ is compatible with the same Pauli basis $P$, we can simultaneously obtain all the expectation value $\tr(\rho Q^{(j)})$ by measuring one basis $P$.
Nevertheless, a practical case generally consists of observables that are compatible with only a subset of other observables.  Then we need to find a set of Pauli bases $\{P\}$ such that each observable $Q^{(j)}$ is compatible with at least one Pauli basis. 
After choosing the Pauli bases~$\{P\}$, the next question is how to distribute the measurement samples to each basis $P$, which corresponds to the idea of importance sampling. 
Without loss of generality, 
we  select each basis $P$ randomly with the probability $\Kcal(P)$.

Now, suppose that we have determined $\{(P, \Kcal(P))\}$, and we can define an estimator
\be
\hat{v} = \sum_{j} \alpha_{j} f(P,Q^{(j)}, \Kcal) \mu\pbra{P, \supp(Q^{(j)})},
\label{eq:UnifiedRep}
\ee
where $\supp(Q):=\cbra{i|Q_i\ne I}$ is the support of $Q$, $\mu(P,\supp(Q^{(j)})) := \prod_{i\in \supp(Q^{(j)})}\mu(P_i)$, and $\mu(P_i)$ is the single-shot outcome by measuring the $i$th qubit with single-qubit Pauli operator $P_i$. 
Here, $\mu\pbra{P,\supp\pbra{Q^{(j)}}}$ effectively gives measurement results of $Q^{(j)}$ obtained from measuring with the basis $P$. 
$f(P,Q^{(j)}, \Kcal)$ is associated with the probability distribution $\Kcal$ of measurement $P$ and $Q^{(j)}$, which is designed to guarantee that $\hat{v}$ is an unbiased estimation of $\tr\pbra{\rho \Omat}$. It depends on the measurement scheme, and we will show its explicit form for different schemes later.
Assuming that $\Ebb\sbra{f\pbra{P, Q^{(j)}, \Kcal}} = 1$, we will show that $\Ebb\sbra{\hat{v}} = \tr\pbra{\rho \Omat}$, i.e., $\hat{v}$ is an unbiased estimator of the observable expectation $\tr\pbra{\rho \Omat}$ in the following proposition.

\begin{proposition}
 Let $\Omat = \sum_j \alpha_j Q^{(j)}$ where $Q^{(j)}$ are local Pauli strings, and $\hat{v}$ be defined as in Eq. \eqref{eq:UnifiedRep} with $\Ebb\sbra{f\pbra{P, Q^{(j)}, \Kcal}} = 1$, then $\hat{v}$ is an unbiased estimation of $\tr\pbra{\rho \Omat}$.
\label{prop:est_unified}
\end{proposition}
\begin{proof}
By the definition of $\hat{v}$, we have
\begin{align}
    \begin{aligned}
    \Ebb\sbra{\hat{v}} &=\Ebb\sbra{\sum_{j} \alpha_{j} f(P,Q^{(j)}, \Kcal) \mu\pbra{P, \supp(Q^{(j)})}}\\
    &= \sum_j \alpha_j \Ebb_{P}\sbra{f\pbra{P,Q^{(j)}, \Kcal}} \Ebb_{\mu|P}\sbra{\mu\pbra{P, \supp\pbra{Q^{(j)}}}}\\
    &= \sum_{j}\alpha_j \tr\pbra{\rho Q^{(j)}} 
    =\tr\pbra{\rho \Omat}
    \end{aligned}
\end{align}
where the first equation holds because of the conditional expectation formula.
\end{proof}

\begin{figure}
    \centering
    \includegraphics[width = 1\textwidth]{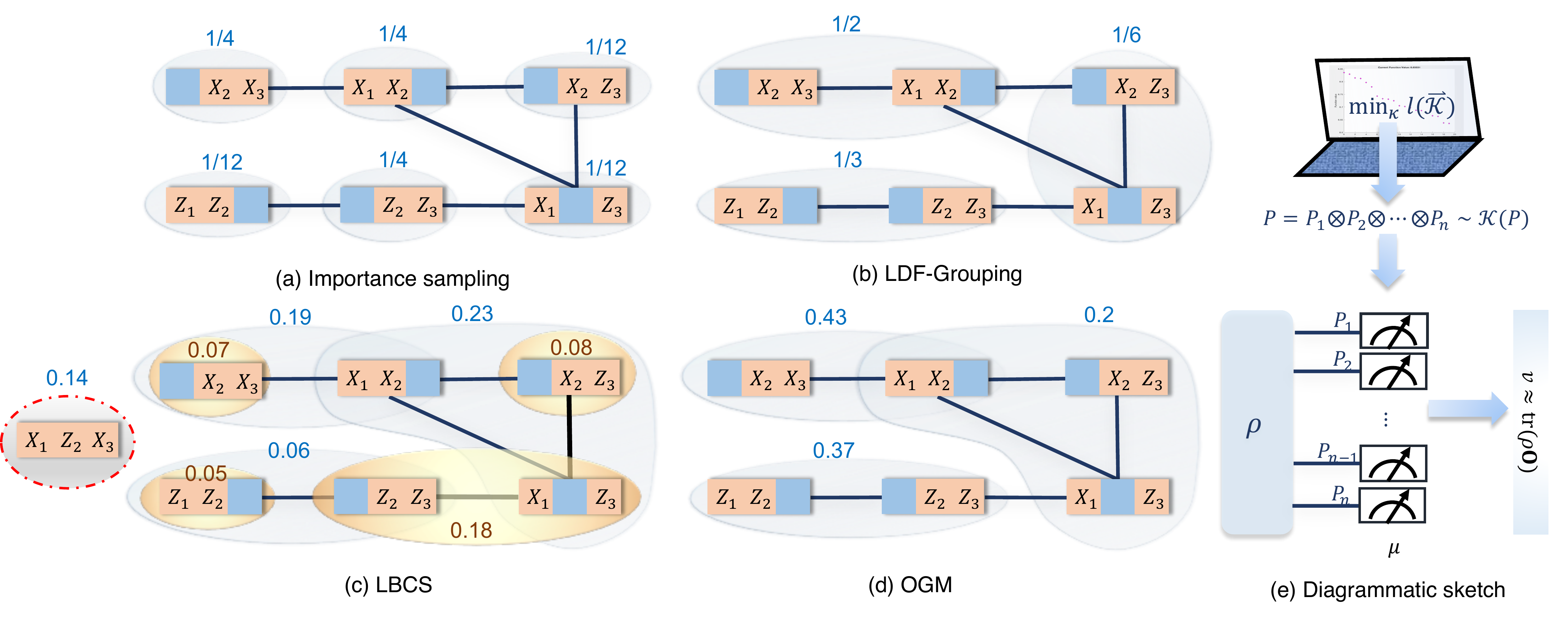}
    \caption{
    An illustrative example of the (a) $l_1$-sampling, (b) LDF-grouping, (c) LBCS, and (d) OGM methods. A vertex represents an observable, and an edge is connected between two vertices when they are compatible. 
    The light blue (yellow) set represents a Pauli measurement basis. Different measurement schemes correspond to different light blue (yellow) sets with different probabilities.
    Here we consider an example of measuring $\Omat = X_1X_2/4 + X_2X_3/4 + X_2Z_3/12 + Z_2Z_3/4 + Z_1 Z_2/12 + X_1Z_3/12$ with a 3-qubit state $\ket{\psi} = \ket{000}/\sqrt{2} + \ket{111}/\sqrt{2}$. The variance for importance sampling,  LDF-grouping, LBCS and OGM algorithms of this instance are $0.90, 0.56, 0.74, 0.50$, respectively.
    (e) Schematic diagram of the unified framework.
    We first determine the set of Pauli basis $\{P\}$ and the probability distributions $\{\Kcal(P)\}$ using  Algorithm~\ref{alg:OverlapSet}.
    We next measure the quantum state $\rho$ with a different Pauli basis drawn from the distribution $\Kcal$ and post-process the measurement outcomes to obtain the estimation of $\tr(\rho\Omat)$.
  }
    \label{fig:GraphComp}
\end{figure}

In the following, we give  explicit expressions of $f(\cdot, \cdot, \cdot)$ for different existing measurement schemes.

(1) \emph{Importance sampling.}
{The strategy of the importance sampling measures each observable $\Qjmat$ independently with the basis $\Pjmat= \Qjmat$,  and the associated probability distribution is determined  by the weight of the observable as $\Kcal(\Pjmat) = |\alpha_j|/\vabs{\bm \alpha}_{1}$ with $\vabs{\bm \alpha}_{1} = \sum_{j =1}^{m} \abs{\alpha_{j}}$ being the $l_{1}$ norm of $\bm \alpha = \pbra{\alpha_{1}, \ldots, \alpha_{m}}$, }
and $f$ defined as
 \be
 f_{l_1}({P, Q^{(j)}, \Kcal}) =  
 \Kcal(P)^{-1} \delta_{P, Q^{(j)} }.
 \ee
It is easy to check that $\Ebb_P\sbra{f_{l_1}\pbra{P,Q^{(j)}, \Kcal}} = 1$.
This method is also referred to as $l_1$-sampling, since the sampling probability is associated with the $l_1$-norm of $Q^{(j)}$.
The $l_1$-sampling needs $\mathcal O({{ \vabs{\bm \alpha}_1^2}/{\varepsilon^2}})$ copies of quantum states $\rho$ to approximate the expectation of $\Omat$ with an additive error
$\varepsilon$. The number of copies is obtained from Chebyshev inequality, and we leave details of error analysis in Appendix~\ref{app:UpperAlgs}.

(2) \emph{Grouping.} The grouping method exploits observable compatibility by partitioning the observables $\mathcal O = \{Q^{(j)}\}$ into several non-overlapped sets $\mathcal S = \cbra{\bm e_1, \ldots, \bm e_s} $ such that $ \bm e_j\cap  \bm e_{j'}=\varnothing ~(\forall j\neq j'$), $\cup_j \bm e_j = \mathcal O$. It also requires that observables in the same set are compatible with each other, such that there exists a measurement basis $\Pjmat$ satisfying $Q \triangleright \Pjmat,\forall Q\in \bm e_j$. Let $\Kcal(\Pjmat)$ be the probability that $\Pjmat$ is selected. It could be optimized using the importance sampling  by setting $\Kcal(\Pjmat)$ proportional to the total weight of the observables in the set $\Pjmat$ as $\Kcal(\Pjmat) = \vabs{\bm e_j}_1/ \vabs{\bm \alpha}_1$. 
Here the weight of a set $\bm e_{j}$ is defined as the $l_{1}$-norm of the weights of the observables in this set as $\vabs{\bm e_j}_1=\sum_{Q^{(k)}\in \bm e_j} \abs{\alpha_k}$. The function $f$ for the optimized grouping method is 
\begin{equation}
    f_{\rm Group}\pbra{\Pjmat, Q^{(k)},\Kcal} = \Kcal(\Pjmat)^{-1} \delta_{Q^{(k)} \in \bm{e}_j}.
\end{equation}
where $\delta_{P,Q^{(j)}}$ equals one if $P = Q^{j}$ and equals zero otherwise. Then we have 
\begin{equation}
    \Ebb_P\sbra{f_{\rm Group} \pbra{P, Q^{(k)}, \Kcal}} = \sum_{j}\Kcal\pbra{P^{(j)}}\Kcal\pbra{P^{(j)}}^{-1}[Q^{(j)}\in \bm e_j] = 1.
\end{equation} 
The number of copies of quantum states $\rho$ needed here is associated with the grouping strategy, and we give an explicit upper bound for the number of copies requiring to approximate $\tr\pbra{\rho\Omat}$ with an additive error
$\varepsilon$ for any given grouping strategy in Appendix~\ref{app:UpperAlgs}.
 
{Finding the exact minimum number of groups has been proved to be NP-hard~\cite{jena2019pauli}}. Several heuristic algorithms, such as the \emph{largest degree first} (LDF) method~\cite{verteletskyi2020measurement} have been proposed to give approximate solutions. We refer to Appendix \ref{app:LDFGroup} for a detailed implementation of the heuristic grouping method.

(3) \emph{Classical shadows.} 
The conventional classical shadow (CS) method measures the quantum state with a random Pauli basis, which corresponds to a Pauli string $P \in \{X,Y,Z\}^{\otimes n}$ within our framework. The original scheme in the seminal work~\cite{huang2020predicting} considers a uniform probability   $\Kcal(P)=1/3^n$ whereas the locally biased classical shadow (LBCS) method~\cite{hadfield2020measurements} assumes a general product distribution $\Kcal(P) = \prod_i \Kcal_i (P_i)$. 
The function $f_{\rm CS}$ is 
\begin{equation}
    f_{\rm CS}\pbra{P, \Qjmat,\Kcal} = \prod_i f_i\pbra{P_i, \Qjmat_i,\Kcal_i},
    \label{eq:cs_f}
\end{equation}
 with $f_i\pbra{P_i, \Qjmat_i,\Kcal_i}= \delta_{Q_i,{I}} + \Kcal_i(P_i)^{-1} \delta_{\Qjmat_i, P_i}$. One can check that \begin{equation}
 \Ebb_{P}\sbra{f_{\rm CS}\pbra{P, \Qjmat, \Kcal}} = \prod_{i} \sum_{P_i}\Kcal\pbra{P_i}f_i\pbra{P_i, \Qjmat_i,\Kcal_i}  = 1.
 \end{equation}
 The number of copies $\rho$ to approximate the expectation of $\Omat$ with an additive error
$\varepsilon$ and $1-\delta$ success probability is bounded to $O\pbra{3^{\text{local(\Omat)}}\pbra{\sum_{j=1}^m \alpha_j}^2/(\delta\varepsilon^2)}$, where local$\pbra{Q^{(j)}} := \# \cbra{k |Q^{(j)}_k \ne \Ibb}$ is the maximum number of qubits $k$ such that $Q^{(j)}_k$ is not identity for all of $j$. We refer to Appendix \ref{app:UpperAlgs} for detailed proof.
 
{(4) \emph{Derandomized CS.} Recently, Huang, Kueng and Preskill~\cite{huang2021efficient} proposed a derandomized classical shadow algorithm, which shows great practical performance compared with conventional classical shadow methods. 
The derandomization algorithm first assigns a collection of $T$ completely random $n$-qubit Pauli measurements, and then derandomizes the process for sampling  measurement set $\mathcal{P}$ by greedily and adaptively choosing current $P^{(j)}$ in the $j$-th step, provided derandomized measurements $P^{(1)}\ldots P^{(j-1)}$, that minimizes the conditional expected value over all remaining random measurement assignments. 
Given all the selected measurements $\mathcal{P}$, the estimator of the derandomized CS algorithm can be expressed as
\begin{equation}
\hat{v} = \sum_{j}\alpha_j\frac{1}{\sum_{P\in\Pcal} \delta_{Q^{(j)}\triangleright P}}\sum_{P:Q^{(j)}\triangleright P}\mu\pbra{P, \supp\pbra{Q^{(j)}}}
    \label{eq:derandCS}
\end{equation}
within our framework. Now we give a sampling version for this method and prove that it can also be unified to the unified framework, as shown in Eq.~\eqref{eq:UnifiedRep}.

For the measurement sequence $P^{(1)},\ldots, P^{(T)}$, suppose the frequency of $P^{(k)}$ be $t_k$, let $p_k = t_k/T$ be the probability to select $P^{(k)}$, and denote this distribution as $\Kcal$. Then proposition \ref{prop:est_unified} still holds as long as for any observable $Q^{(j)}$ where $j\in[m]$, there exists a measurement $P^{(k)}$ in the measurement sequence such that $Q^{(j)}\triangleright P^{(k)}$.
In this case, we can rewrite Eq. \eqref{eq:derandCS} as
\begin{equation}
    \hat{v} = \sum_{j}\alpha_j \frac{1}{p\pbra{Q^{(j)}}} \mu\pbra{P^{(k)},\supp\pbra{Q^{(j)}}}
    \label{eq:revise_derand}
\end{equation}
for a selected measurement $P^{(k)}$, where $p\pbra{Q^{(j)}}$ is the probability to measure $Q^{(j)}$. It is easy to check that $\Ebb\sbra{f\pbra{P^{(k)},Q^{(j)}, \Kcal}} = \Ebb\sbra{ \frac{1}{p\pbra{Q^{(j)}}}} = 1$, and hence 
$\Ebb\sbra{\hat{v}} = \tr\pbra{\rho \Omat}$.
}
 
In Fig.~\ref{fig:GraphComp}, we show explicit examples of the above three typical methods. While Ref.~\cite{hadfield2020measurements} showed the superiority of the LBCS method, we can see that LBCS essentially exploits an alternative view of observable compatibility, which is captured by the unified framework. However, since the measurements are selected locally in LBCS, we have to measure redundant observables, such as $X_1Z_2X_3$ in our example {of Fig. \ref{fig:GraphComp}}. This term has no contribution to the objective observable, but is still assigned a certain number of measurements. For a general observable, many  measurements might be assigned to these redundant terms, and thus makes LBCS non-optimal or even costly with increasing system size.

\section{Overlapped grouping measurement}
\label{sec:overlapped_group}
Here, we propose a new scheme that
exploits the advantages of the aforementioned typical measurement methods. 
We first introduce the concept of \emph{overlapped grouping} and then give a comparison of this strategy and other existing strategies.

\subsection{Overlapped grouping}

\begin{definition}[Overlapped grouping]
For a set of observables $\mathcal O=\{Q^{(j)}\}$, the collection $\mathcal S = \{\bm e_1, \ldots, \bm e_s\}$ is an overlapped grouping when  $\cup_j \bm e_j = \mathcal O$ with corresponding measurements $\{P^{(1)}, \ldots, P^{(s)}\}$ satisfying $Q \triangleright \Pjmat,\forall Q\in \bm e_j$.
\label{def:OverlapGroup}
\end{definition}
\noindent 
Suppose that we have determined the 
probabilities $\{\Kcal(P^{(j)})\}$, and we can define a new function $f$ for the overlapped grouping as
\begin{equation}
    f_G(P,Q, {\Kcal}) =  \chi(Q)^{-1} \delta_{Q \triangleright P},
\label{eq:DefFuncOGM}
\end{equation}
{where $\chi(Q) := \sum_{P:Q \triangleright P} \Kcal(P)$ represents the probability that $Q$ is effectively measured with the basis $P$.
Now we can define}
\begin{equation}\label{Eq:vgdef}
    \hat{v}_G = \sum_{j} \alpha_{j} f_G(P,Q^{(j)}, \Kcal) \mu\pbra{P, \supp(Q^{(j)})}
\end{equation}
as an unbiased estimator of $\tr\pbra{\rho \Omat}$. Intuitively, an unbiased estimation of $\tr\pbra{\rho Q^{(j)}}$ can be generated from the measured results of $Q^{(j)}$ divided by its measured probability.
From the definition of $f_G$,  we have $
         \Ebb_P[f_G(P,Q,\Kcal)] = \sum_{P:Q\triangleright P}\Kcal(P)\frac{1}{\chi(Q)}=1$, and
hence $\hat{v}_G$ is also an unbiased estimation of $\tr\pbra{\rho \Omat}$ by Proposition~\ref{prop:est_unified}. To summarise, an overlapped grouping measurement (OGM) scheme works as follows.

\begin{itemize}
    \item [S1.] Find the overlapped sets $\mathcal S$  with corresponding measurements $\{P^{(1)}, \ldots, P^{(s)}\}$.
    \item [S2.] Find the probability distribution $\{\Kcal(\Pjmat)\}$.
    \item [S3.] Measure the quantum state with a randomly generated basis $P$ and process the outcomes with Eq.~\eqref{eq:UnifiedRep} and~\eqref{eq:DefFuncOGM}. 
\end{itemize}

A specific OGM scheme is determined by the choice of sets $\mathcal{S}$ (equivalently $\{P^{(j)}\}$) and the
probability distribution $\{\Kcal(P^{(j)})\}$.
To quantify the performance of the scheme, we consider the variance of the estimator, as shown in the following proposition. 
\begin{proposition}
The variance of $\hat{v}_G$ defined in Eq.~\eqref{Eq:vgdef} is
\begin{equation}
    \mathrm{Var}(\hat{v}_G)=\sum_{j,k} \alpha_j\alpha_k g(Q^{(j)},Q^{(k)})\tr\pbra{\rho Q^{(j)}Q^{(k)}}{-\tr\pbra{\rho\Omat}^2}
    \label{eq:VarOGM}
\end{equation}
where $g(Q^{(j)},Q^{(k)}) := \chi^{-1}(Q^{(j)}) \chi^{-1}(Q^{(k)}) {\sum_{P: Q^{(j)}\triangleright P \land Q^{(k)} \triangleright P} \Kcal(P)}{}$.
\label{prop:OverlapGrouping}
\end{proposition}

\begin{proof}
Since Var$(\hat{v}_G) = \Ebb\sbra{\hat{v}_G^2} - \Ebb\sbra{\hat{v}_G}^2$, Proposition \ref{prop:OverlapGrouping} follows from the following equations
\begin{equation}\label{eq:VarMu}
    \begin{aligned}
    &\Ebb_{P} [f_G(P,Q,\Kcal) f_G(P,R,\Kcal)] = \frac{\sum_{P:Q\triangleright P \land R\triangleright P} \Kcal(P)}{\chi(Q)\chi (R)},\\
    &\Ebb_{\mu(P|Q,R)}[\mu(P, \supp(Q))\mu(P,\supp(R))] = \tr(\rho QR),
\end{aligned}
\end{equation}
where the first equation holds directly by its definition and the second equation holds because we have   $\mu(P,\supp(Q))\mu(P,\supp(R)) = \mu(P,\supp(Q\oplus R)) = \mu(P,\supp(QR))$ when   $Q,R\in \{I,X,Y,Z\}^n$. 
\end{proof}
The variance determines the sample complexity. In particular, we need the total number of measurements $T\geq \text{Var}\pbra{\hat{v}_G}/\pbra{\delta\varepsilon^2}$ samples to achieve $\Pr\sbra{\abs{\hat{v}_G - \tr\pbra{\rho \Omat}}\geq \varepsilon}\leq \delta$ with error $\varepsilon\ge 0$ and failure probability $\delta\in[0,1]$.

\subsection{Illustration and comparison with other measurement schemes}
\label{subsec:compOGM}

We illustrate the differences between our OGM scheme and other measurement schemes in Fig.~\ref{fig:GraphComp}. As illustrated in Fig. \ref{fig:GraphComp}, importance sampling selects an observable in each iteration, and measures the prepared state with the sampled observables to obtain the estimations associated with this observable. Grouping strategy leverages the compatible property of the observables, and measures the observables that are compatible jointly. Nevertheless, it only exploits a very limited space of the full probability space for $4^n$ possible measurements in $\cbra{I,X,Y,Z}^n$. Moreover, for the sake of the grouping determination  using the heuristic strategy, an observable can not arise in two different sets. Therefore, it might be inefficient in leveraging of the measurements. Classical shadow method finds the optimized probabilities of each qubit of the measurement, and it also measures the observables jointly. However, since the CS method independently generates Pauli operators on each qubit, it will generate useless measurements, such as $X_1Z_2X_3$ in Fig.~\ref{fig:GraphComp}(c). 
As a comparison, for any measurement $P$ in set $\cbra{P^{(1)},\ldots, P^{(s)}}$ generated from the OGM scheme, there exists at least an observable $Q^{(j)}$, such that $Q^{(j)}\triangleright P$.

The overlapped grouping measurement framework defined as in Eq. \eqref{eq:DefFuncOGM} without an explicit assignment of $\cbra{P^{(k)}}$ and $\Kcal$ covers importance sampling, LDF Grouping, CS and the ``probabilistic version'' of the {derandomized CS} algorithm as in Eq. \eqref{eq:revise_derand}.
The importance sampling, grouping and CS algorithms can be regarded as a special OGM framework with some restrictions for the distribution of measurements $\{(P,\Kcal(P))\}$.
We note that in our method $f_G(P,Q, {\Kcal}) =  \chi(Q)^{-1} \delta_{Q \triangleright P}$ may lead to more effective data post-processing since it exploits all the compatible properties of observables in the overlapped sets, and the ``probabilistic version'' estimation in Eq. \eqref{eq:revise_derand} of the {derandomized CS} algorithm also belongs to this scope, since it is equivalent to the estimation expression with the OGM grouping strategy in Eq. \eqref{Eq:vgdef}.

We can also observe that OGM is strictly better than the classical shadow method.
The OGM scheme reduces to the CS method when we choose $\Pjmat\in \{X,Y,Z\}^{\otimes n}$ and restrict 
the probability distribution $\Kcal(\Pjmat)$ 
to have a local product structure on different qubits. We remark that OGM will not measure redundant observables as that in the local shadow methods.

The estimator in Eq.~\eqref{eq:derandCS} indicates that derandomization utilizes the compatible properties of observables when measuring on the predetermined basis. Once the measurement bases are determined, it could be regarded as a special overlapped grouping method.

In OGM, challenges remain to (1) determine the collection $\mathcal S$ and (2) find the probability distributions $\Kcal$. Similar to the case of grouping method, finding the optimal overlapped groups given the objective observables is also NP-hard.
{In Sec.~\ref{sec:explicit_strategy}, we develop an explicit strategy to determine an approximate solution $\mathcal S$
by leveraging a greedy algorithm based on the weights of the observables. }
To find the probability distribution $\Kcal$, we apply an optimization procedure to adaptively search for the solution that minimizes the estimator variance.

\section{Explicit grouping strategies}
\label{sec:explicit_strategy}
We show in Algorithm~\ref{alg:OverlapSet} our strategy to determine the overlapped sets $\bm e_1, \ldots, \bm e_s$ and the associated probability $\Kcal_j:=\Kcal(P^{(j)})$. The main idea is that under the premise of covering all the objective observables{, we} add an observable which has not been accessed into a new set, and add all compatible observables into this set. 
We give priority to observables with larger absolute weights since it has more contributions {to} the estimation. We note that different sequences to add a new observable into an existing set will influence the structure of sets and the number of sets. Algorithm~\ref{alg:OverlapSet} provides a grouping strategy by adding a new observable by its importance (weight) and trying to reduce the number of sets as far as possible.
Meanwhile, the procedure guarantees that whenever an observable $Q^{(j)}$ is compatible with the measurement $P^{(i)}$, it is in the set $\bm e_i$. See Appendix \ref{app:OGMGroup} for an alternative strategy, which has slightly better performance however based on a more dedicated optimization procedure.

\SetNlSty{textbf}{}{\quad}
\IncMargin{1em}
\begin{algorithm}
\SetKw{KwTo}{to}
\SetKwFor{While}{while}{:}{}
\SetKwFor{For}{for}{:}{}
\SetKwIF{If}{ElseIf}{Else}{if}{:}{elif}{else:}{}
\SetKwInOut{Input}{Input}
\SetKwInOut{Output}{Output}
\Input{$n,m$ and $Q^{(1)},\ldots, Q^{(m)}, \alpha_1,\ldots, \alpha_m$.}
\Output{$\{P^{(s)}\}$ with initial probabilities $\{\Kcal_s\}$.}
\emph{Sorting $\{Q^{(j)}\}$ to the descending order according to their weights $\abs{\alpha_j}$}\;
$j\leftarrow 1$ and $s\leftarrow 1$\;
\While{$\exists Q^{(j)}$ that is not in any sets}{
\qquad\emph{ Let $Q^{(j)}$ be the first observable in the sorted sequence which has never appeared in any sets and add it into a new set $\bm e_s$}\;
\qquad\emph{$s\leftarrow s+1$}\;
\qquad\emph{Initialize the measurement of $\bm e_s$ as $P^{(s)}\leftarrow Q^{(j)}$}\;
\qquad\For{$k\leftarrow j+1$ \KwTo $m$ }{

\qquad\qquad\If{observable $Q^{(k)}$ is compatible with $P^{(s)}$}{
\qquad\qquad\qquad\emph{Add $Q^{(k)}$ into set $\bm e_s$}\;
\qquad\qquad\qquad\emph{Update $P^{(s)}$ to $ P^{(s)} \bigvee Q^{(k)}$}
\tcc*{$P = Q\bigvee R$ is defined as  $P_i = Q_i$ if $Q_i = R_i$ and $P_i = Q_iR_i$ otherwise.}
    }
}
\qquad\emph{Let the initial probability of $P^{(s)}$ be the summation of the weight of all observables in this set}\;
\qquad\For{$k\leftarrow 1$ \KwTo $j-1$}{
\qquad\qquad\If{observable $Q^{(k)}$ is compatible with $P^{(s)}$}{
\qquad\qquad\qquad\emph{Add $Q^{(k)}$ into set $\bm e_s$, and update $P^{(s)}$ to $ P^{(s)} \bigvee Q^{(k)}$}\;
}
}
}
\caption{Overlapped set generation.
}
\label{alg:OverlapSet}
\end{algorithm}
\DecMargin{1em}

{The algorithm outputs the measurements $\{P\}$ with non-optimized probabilities $\{\Kcal\}$. 
Here the initial probability of $P^{(s)}$ is not chosen as the weight of $\bm e_{s}$ since we wish to distinguish the importance of different sets and give more priority to the sets which are generated in front of $\bm e_s$.
We can then optimize $\{\Kcal\}$ to further minimize the estimator variance. 
However, the variance $\textrm{Var}(\hat{v}_G)$ in Eq.~\eqref{eq:VarOGM} depends on the input state $\rho$, which could be unknown in general.
Alternatively, we consider the diagonal approximation of} {Var$\pbra{\hat{v}_G}$} (see similar techniques in Ref. \cite{hadfield2020measurements}), which is explicitly expressed as
\begin{equation}
    l(\vec \Kcal)= \sum_{j} \frac{\alpha_j^2}{\chi({Q^{(j)}, \vec \Kcal})},
    \label{eq:diag_var}
\end{equation}
where $\chi(Q, \vec \Kcal) = \sum_{P:Q \triangleright P} \Kcal(P)$ and $\vec \Kcal := \pbra{\Kcal_1, \cdots, \Kcal_s}$ represents all the corresponding probabilities. 
We give the mathematical supports that why we utilize $l(\vec \Kcal)$ as the cost function in Appendix \ref{app:CFandVariance}.
There are several advantages of using the diagonal approximation $l(\vec \Kcal)$ instead of the actual variance~---~(1) independence of the quantum state, (2) fast classical evaluation, (3) including dominant contribution to the variance since $\tr\pbra{\rho Q^{(j)}Q^{(k)}}< \tr\pbra{\rho Q^{(j)}Q^{(j)}} = 1$ when $j\ne k$.  Therefore, we could instead regard $l(\vec \Kcal)$ as the cost function and minimize it by optimizing over $\vec \Kcal$. 
From the expression of $l\pbra{\vec \Kcal}$ in Eq.~\eqref{eq:diag_var}, we see the cost function is not convex in $\vec \Kcal$ and hence there is no closed minimum solution.
An estimation can be generated by searching for a local minimum solution of the cost function in Eq.~\eqref{eq:diag_var}. To further give a better estimation and avoid being trapped into bad local minima, we slightly revise the cost function, as shown in the following subsection.

\subsection{Optimization process}
For the optimization process of the OGM method, 
we will further speed it up by adaptively deleting the groups that have very small initial probabilities until the cost function stops decreasing with the disturbance.
Note that after cutting down the groups with small weights, some observables with small coefficients will disappear in the cost function. Therefore, we adjust the final cost function as
\be
l(\vec \Kcal) = \sum_{Q^{(j)}\in \Scal} \frac{\alpha_j^2}{\chi\pbra{Q^{(j)}}} + \sum_{Q^{(j)}\not\in \Scal }\alpha_j^2T,
\label{eq:lfFinal}
\ee
where $T$ is the total number of samples, $Q^{(j)}\in \Scal$ if there exists a set $\bm e$ such that $Q^{(j)}\in \bm e$, and $\sum_{Q^{(j)}\not\in \Scal}$ is the penalty caused by deleting some sets. 
The selection of the final cost function in Eq. \eqref{eq:lfFinal} is inspired by the relationship between the variance and the number of samples. 
More specifically, Chebyshev inequality $T\geq \text{Var}(\hat{v})/\pbra{\delta\varepsilon^2}$ indicates that Var$(\hat{v})$ is linear in $T$. Hence we introduce $\alpha_j^2 T$ to compensate the initial error $\varepsilon_0$ for excluding the observable $Q^{(j)}$.
The initial error $\varepsilon_0 = \abs{\sum_{j:Q^{(j)}\not \in \Scal} \alpha_j \tr\pbra{\rho Q^{(j)}}}$ implies biases of our estimation.
We could search for an optimized $T$ in a real experiment with a small-scaled input size with an initial $T_0$.
Since the cost function in Eq.~\eqref{eq:lfFinal} is not convex, we could find a local minimum solution using the nonconvex optimization methods. 

Since our OGM method assumes measurements drawn from the probability distribution, the measurement accuracy may fluctuate. We will derandomize the scheme by fixing almost all of the choices of measurements $P$ in the next subsection. 

\subsection{Sampling strategy}

Suppose that we have determined the measurement basis sets $\bm e_1, \ldots, \bm e_s$ and the optimized probability distribution $\{(P,\Kcal (P))\}$ using the above strategy. In practical computation, we usually have constraints on the maximum allowed number of measurements. In what follows, we provide a partially derandomized strategy with the given number of measurements $T$.
For the $j$th measurement $P^{(j)}$ with sampling probability $\Kcal_j$, we choose $\floor{T \Kcal_j}$ number of measurements for $P^{(j)}$, and select an additional one $P^{(j)}$ with probability $T\Kcal_j - \floor{T\Kcal_j}$, as shown in the Algorithm~\ref{alg:ParDerand}. 

\IncMargin{1em}
\begin{algorithm}[h]
\SetKw{KwTo}{to}
\SetKwFor{For}{for}{}{}
\SetKwInOut{Input}{Input}
\SetKwInOut{Output}{Output}
\Input{Measurements $ P^{(1)},\ldots, P^{(s)}$, which follows distribution $\Kcal = \cbra{\Kcal_1, \ldots, \Kcal_s}$, the number of selected measurements $T$.}
\Output{List $\Mcal$ which contains selected measurements.}
\emph{Sorting measurements $P^{(1)},\ldots, P^{(s)}$ to descending order in terms of its probabilities}\;
\For{$j\leftarrow 1$ \KwTo $s$}{
\qquad \emph{Add $\floor{\Kcal_jT}$ number of $P^{(j)}$ into list $\Mcal$}\;
\qquad \emph{$\Kcal_j\leftarrow \Kcal_j T - \floor{\Kcal_jT}$}\;
}
\For{$j\leftarrow 1$ \KwTo $s$ and size$(\Mcal)<T$}{
\qquad \emph{Add one $P^{(j)}$ into list $\Mcal$ with probability $\Kcal_j$}\;
}
\Return{$\Mcal$}\;
\caption{{Partially derandomized sampling process for OGM.}}
\label{alg:ParDerand}
\end{algorithm}
\DecMargin{1em}
Observe that the estimation $\hat{v}_G$ does not rely on the arrangement of measurements. Let $\Mcal$ be the list for the selected measurements, and note here we allow a measurement to appear more than once in the list $\Mcal$. 
It is easy to check the expectation number of samples for $P^{(j)}$ is equal to $\Kcal_jT$ with Algorithm \ref{alg:ParDerand}.
Note that in Algorithm~\ref{alg:ParDerand}, the number of sampled measurements may not exactly equal $T$, although it is close to $T$ if $s\ll T$. Hence in the numerical experiment we additionally add $P^{(j)}$ to $\Mcal$ for the measurement $P^{(j)}$ satisfies $\Kcal_jT<1$ and size$(\Mcal)< T$ in the descending sequence sorted in Step (1) of Algorithm~\ref{alg:ParDerand} if the size of $\Mcal$ is less than $T$. We provide detailed discussions on the variance of the partially derandomized strategy in Appendix~\ref{appendix:derand}.

\section{Numerical tests}
\label{sec:numerical_test}
In this section, we numerically demonstrate the overlapped grouping measurement algorithms for the energy estimation of molecular systems, and compare our methods with other advanced measurement strategies, including LDF-grouping, locally biased classical shadows, and derandomized classical shadows. {We do not include importance sampling method in the comparison since its performance is worse than others.}
{Algorithm \ref{alg:TotalOGM} gives the full estimation process for the OGM algorithm.}
\IncMargin{1em}
\begin{algorithm}[h]
\SetKw{KwTo}{to}
\SetKwFor{For}{for}{}{}
\SetKwInOut{Input}{Input}
\SetKwInOut{Output}{Output}
\Input{Hamiltonian $\Omat = \sum_i \alpha_i Q^{(i)}$, the number of samples $T$.}
\Output{Estimation $v$.}
\emph{Generate the initial measurement distribution $\cbra{P^{(s)}, \Kcal_s}$ with Algorithm \ref{alg:OverlapSet} or Algorithm \ref{alg:OverlapSetV2}}
\label{step:first_ogm}\;
\emph{Uniformly randomly pick $A = 10$ groups of distributions $\cbra{\Kcal_s}$ proportional to $[P^{(s)}, P^{(s)} + \max_s P^{(s)}]$, and choose one with minimum cost function Eq. \eqref{eq:lfFinal} as the initial point}
\label{step:sec_ogm}\;
\emph{Find the local optimal distribution $\cbra{P^{(s)}, \Kcal_s}$ by optimizing the cost function Eq. \eqref{eq:lfFinal} with Matlab optimization package}\;
\emph{Sample $T$ measurements from the optimized distribution with Algorithm \ref{alg:ParDerand}}\;
\emph{Calculate $v$ with Eq. \eqref{eq:derandCS}}\;
\caption{{Overlapped grouping measurement.}}
\label{alg:TotalOGM}
\end{algorithm}
\DecMargin{1em}
{In Step \ref{step:sec_ogm} of Algorithm \ref{alg:TotalOGM}, we begin the optimization process from a better initialized probabilities by picking the distribution with the minimum cost function from uniformly randomly selected 10 distributions around the initialized distribution $\Kcal_s$ from Step \ref{step:first_ogm}. To show the robust advantages of the OGM algorithm, we directly choose the  probabilities initialized in Algorithm \ref{alg:OverlapSet} without performing Step \ref{step:sec_ogm} to give the optimized measurement distributions, and outputs the errors for estimations with 1000 samples in Table \ref{tab:ErrorComp_main}.}

We compare the measurement schemes for different molecular Hamiltonians, ranging from $4$ to $16$ qubits. We first consider the molecular Hamiltonian measurement on the ground state of molecular Hamiltonians,  in which the fermionic Hamiltonians are mapped to the qubit ones under the Jordan-Wigner (JW) transformation and
the number of terms in the molecular Hamiltonians scales quartically to the system size. In practice,  the cost function in Eq.~\eqref{eq:diag_var} might lead to the optimized result for the probability distributions trapped in the local minimum. Here, we address this problem by adding an additional disturbance term in the cost function to jump out of local minima.

We compare the estimation error (averaged over 100 independent tests) using $1000$ measurement samples in Table \ref{tab:ErrorComp_main}. {Here the error is estimated with the formula $
\varepsilon_v = \sqrt{\frac{1}{N}\sum_{i = 1}^N \pbra{\hat{v}_i - \tr(\rho O)}^2}.$ The definition of $\varepsilon_v$ is also consistent with the standard deviation calculation of the estimation $v$.} It is worth mentioning that we numerically show that $N=1000$ independent experiments are sufficient to output a convinced estimation error in Appendix \ref{app:numericalError}. 
We also include the recently proposed derandomized classical shadow method, which is the current state-of-the-art method and has been numerically tested to outperform the others~\cite{huang2021efficient}. {The numerical result again shows that our OGM method achieves much higher accuracy than other methods {when the number of measurements is limited}, including the derandomized classical shadow method, verifying its significant performance in the practical computation.
The OGM algorithm has simultaneous advantages in the energy estimation under different fermion-to-qubit encodings, including  Bravyi-Kitaev (BK) and parity encodings, and we refer to Appendix \ref{app:numericalError} for the numerical results and detailed comparison under Bravyi-Kitaev and parity encodings.
}

\begin{table}[htbp]
    \centering
        \caption{{Estimation errors with 1000 samples for measuring molecular Hamiltonians (JW-encoding) on the corresponding ground states via LDF-grouping, LBCS, derandomized CS, and the OGM method, {where the initialized probability for the OGM algorithm is chosen directly from Algorithm~\ref{alg:OverlapSet}}.}}
    \begin{tabular}{c| c c c c }
    \hline
       Molecule & LDF & LBCS~\cite{hadfield2020measurements} & Derand~\cite{huang2021efficient} & OGM\\
       \hline\hline
        H$_2$($4$)  & 0.019 &0.043 & 0.018 & \textbf{0.011}\\
        \hline
        H$_2$ ($8$) & 0.149 &0.128 & 0.067 &  \textbf{0.051}\\
        \hline
        LiH ($12$)  & 0.231 & 0.122 & 0.063 &  \textbf{0.036}\\
        \hline
        BeH$_2$ ($14$)  & 0.426 & 0.275 & 0.103 &  \textbf{0.072}\\
        \hline
        H$_2$O ($14$)  & 1.090 & 0.549 & 0.257 & \textbf{0.129 }\\
        \hline
        NH$_3$ ($16$)  & 1.063 & 0.484 & 0.225 &  \textbf{0.151}\\
        \hline
    \end{tabular}
    \label{tab:ErrorComp_main}
    
\end{table}

To verify that the OGM method also has advantages when we have a large number of measurements, we show the comparison of the OGM algorithm and LDF Grouping, LBCS, and derandomized CS algorithm for errors with different  number of measurements for molecules LiH, and BeH$_2$ under the JW encoding, as shown in Fig.~\ref{fig:err_molecules}.
\begin{figure}[t]
    \centering
    \includegraphics[trim = 0mm 46mm 0mm 38mm, clip=true, width = 1.05\textwidth]{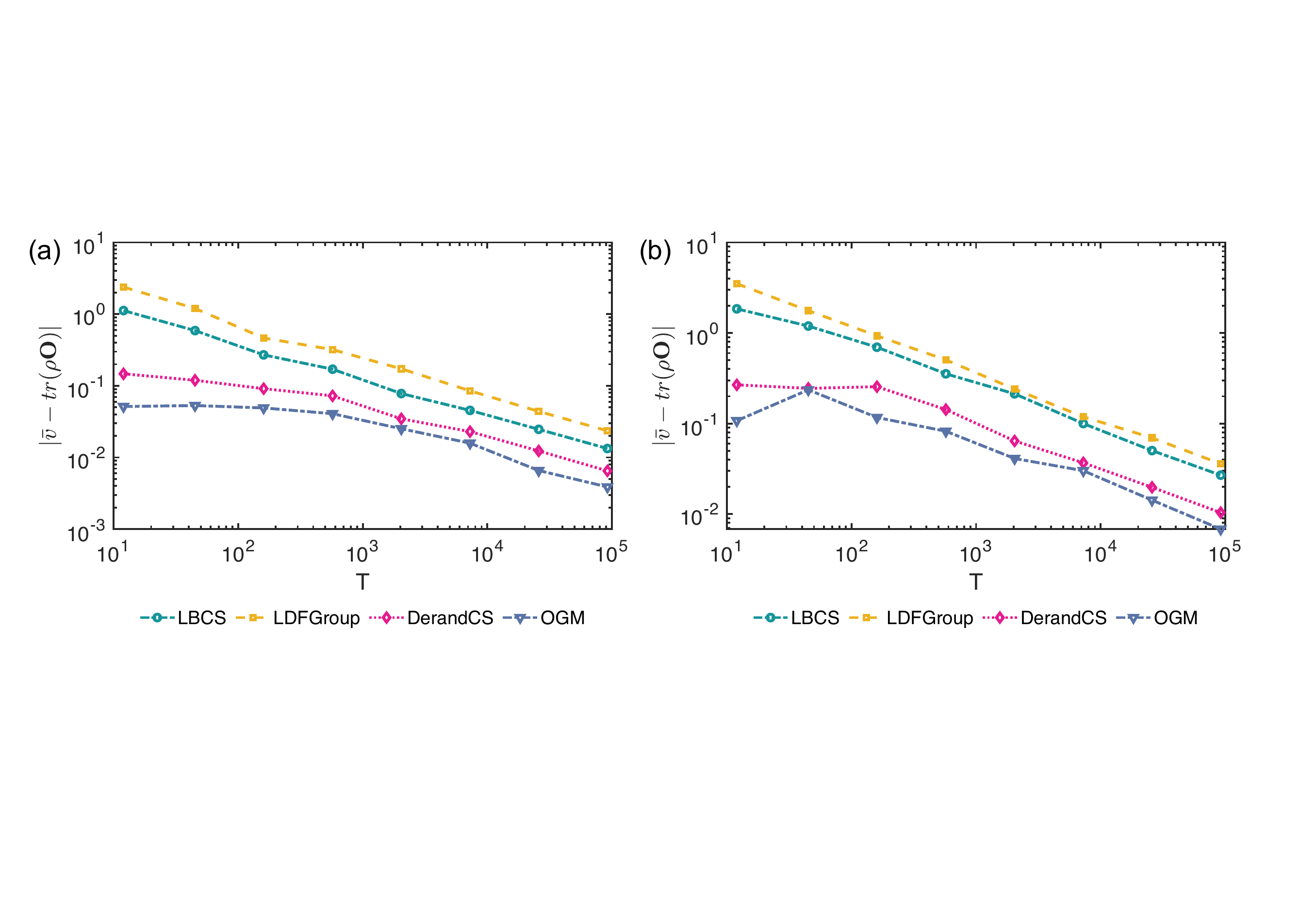}
    \caption{{The error of the estimations by the LBCS, LDF Grouping, derandomized CS, and OGM algorithm with the number of measurements for the ground state energy of molecules (a)  LiH (12 qubits), and (b) BeH$_2$ (14 qubits)  with JW-encoding.}}
    \label{fig:err_molecules}
\end{figure}
The results show that the OGM algorithm scales linearly to $1/\sqrt{T}$ for a large number of samples $T$. This numerical result is consistent with the theoretical results. 
It can be shown that the OGM algorithm has clear advantages for a large number of samples compared to other existing algorithms. We also provide the comparison of the variances of existing algorithms with OGM algorithm in Appendix \ref{app:variance_comp}.
   
\begin{figure}[htbp]
        \centering
      \includegraphics[trim = 0mm 55mm 0mm 42mm, clip=true, width = 1.0\textwidth]{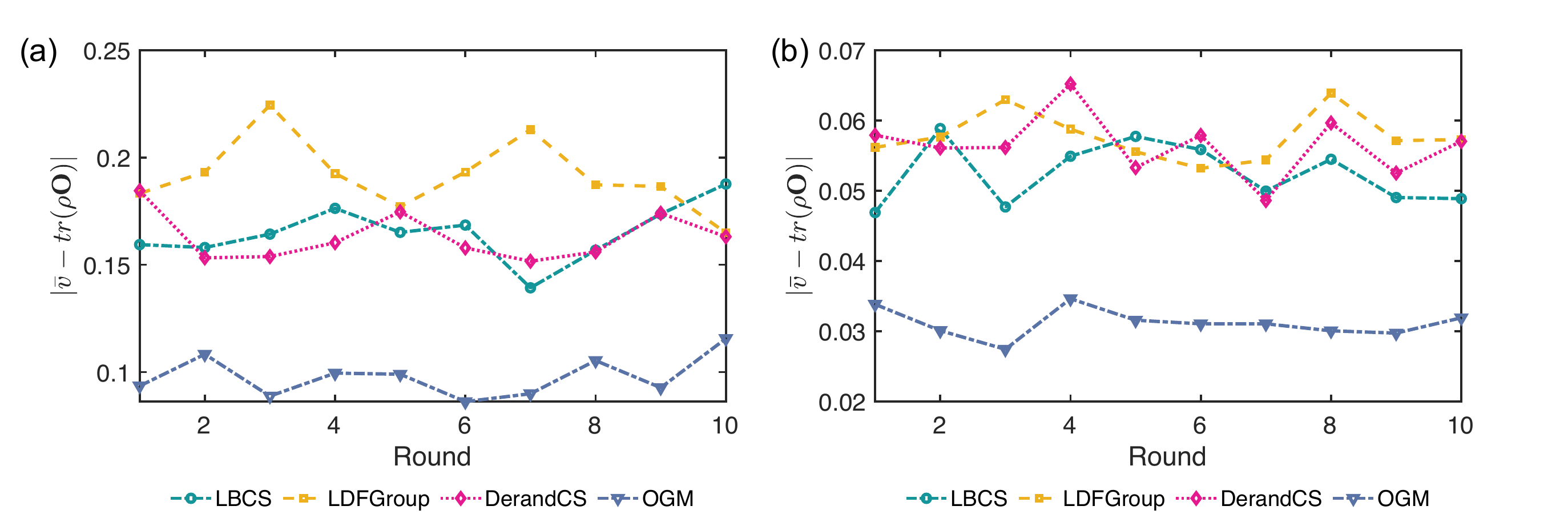}
        \caption{The error of the estimation by the LBCS, LDF Grouping, Derandomized CS and OGM algorithm for the expectation of molecule $H_2$ (8 qubits) with JW-encoding under 10 random generated 8-qubit state. (a) $T = 1000$ number of samples and (b) $T=10,000$ number of samples.}
        \label{fig:randH2}
    \end{figure}

{We also show Fig.~\ref{fig:randH2} and \ref{fig:error_rand} to illustrate that the advantage of OGM is independent of the quantum input state, where we approximate the expectation of molecule H$_2$ (8 qubits) with JW-encoding. In Fig.~\ref{fig:randH2}, we compare the errors of LBCS, LDF Grouping, Derandomized CS, and OGM algorithms under 10 random generated 8-qubit states with (a) 1000 samples and (b) 10,000 samples. In Fig.~\ref{fig:error_rand} we further show the comparison of these algorithms on a  randomly generated 8-qubit state with the increase of the number of samples, where the $x$-axis and the $y$-axis are both in logarithmic scales.
Here we choose the input quantum state as an $8$-qubit state with uniformly randomly generated real amplitudes.
}

{We additionally provide the experimental results in Appendix \ref{app:experimental_test}. The experimental results clearly show a much faster convergence of our OGM method using a few hundred of measurements, which aligns with our theoretical prediction and numerical simulation.
We can observe that our methods are practically useful even for the current generation of quantum devices.}

\section{Discussion and outlook}
\label{sec:discussion}
We  introduce a unified framework of quantum measurement that reveals the underlying mechanism of the existing advanced measurement strategies, which are seemingly distinct from each other. We further propose the overlapped grouping measurement (OGM) scheme that integrates the advantages of these typical measurement strategies.
Our numerical results suggest a significant improvement over existing advanced measurement methods. Our numerical result shows that our method already demonstrates advantages in practical problems. 
{Since the efficient quantum measurement is crucial for many quantum algorithms and quantum processing,} our work has wide applications, such as in variational quantum algorithms and quantum many-body tasks involving eigenenergy estimation~\cite{mcardle2018quantum,Cao2019Quantum,babbush2018low,arute2020hartree}, where we need to efficiently measure complicated Hamiltonians $\braket{H}$ or their moments~$\braket{H^2}$~\cite{vallury2020quantum}.
{Our method could significantly reduce the measurement cost and hence speed up the quantum computation, especially when we aim to realize quantum advantage for realistic problems.}
Moreover, our method applies to adaptive variational quantum simulation, which requires a large number of measurements in each subroutine~\cite{zhang2020low,grimsley2019adaptive}. It is expected that our measurement scheme will show more advantages with an increasing system size of great practical relevance to both theoretical and experimental tasks.

 \begin{figure}[htbp]
    \centering
    \includegraphics[trim = 0mm 85mm 0mm 72mm, clip=true, width = 0.55\textwidth]{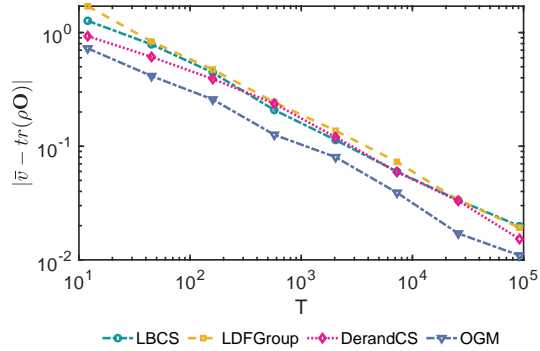}
    \caption{The error of estimation by OGM algorithm for the expectation of molecule $H_2$ (8 qubits) with JW-encoding under a randomly generated 8-qubit state.}
    \label{fig:error_rand}
\end{figure}

{The optimization goal of the OGM algorithm is completely different from the derandomized CS method, since here we utilize a partial variance as the cost function, while the derandomized CS algorithm utilizes a confidence bound. The numerical results also show that our algorithm has clear advantages for a large number of measurements. Our work considers explicit strategies for choosing the overlapped sets, which could be improved using more advanced classical algorithms.
Note that the dimension of the considered measurement space is $4^n$, and both OGM and CS variant algorithms aim to find a good distribution in this huge space. The expressivity of the CS algorithm is limited since it only explores the $3n$ size of this space.
One of the approaches is to combine OGM algorithm with CS for the molecules where parts of qubits have strong correspondence. 
Briefly speaking, we could utilize the CS method to generate $n/s$ independent subspaces, each grouped with OGM with dimension $4^S$. We leave this idea as an interesting future work. Another possible extension is to utilize neural networks to generate  samples within the OGM framework~\cite{torlai2018neural,carrasquilla2019reconstructing,torlai2020precise}}.
In our work, we assume local Pauli measurements, whereas more general measurements, such as arbitrary local measurements or entangled measurements could be considered~\cite{mcclean2016theory,izmaylov2019revising,rubin2018application,o2019generalized,Yen2019,izmaylov2019unitary,zhao2020measurement}. Several measurement schemes have been proposed by adding a polynomial-depth circuit before the local measurement to implement entangled measurements~\cite{huggins2019efficient,Yen2019,Gokhale2019}. How to extend our OGM scheme to generalized measurements is an interesting future direction. 

\section*{Acknowledgement}
We would like to thank Charles Hadfield and Antonio Mezzacapo for providing the Hamiltonians with their work~\cite{hadfield2020measurements} and sharing relevant recent papers. We thank Charles Hadfield, Antonio Mezzacapo, Xiaoming Sun, and Xiaoming Zhang for their helpful discussions.
This work is supported by the National Natural Science Foundation of China (Grant No.~12175003, No.~12147133), and Zhejiang Lab's International
Talent Fund for Young Professionals.
The numerics is supported by High-performance Computing Platform of Peking University.
We acknowledge the use of the IBMQ for
this work. The views expressed are those of the authors
and do not reflect the official policy or position of IBM
or the IBM Q team\footnote{The source code for the OGM optimization process is available at https://github.com/GillianOoO/Overlapped-grouping-measurement.}.

\emph{Note added.---}Recently, two relevant works~\cite{hadfield2021adaptive, hillmich2021decision} were posted, which introduce optimized quantum measurement schemes that generalize the classical shadow methods. Hillmich et al.~\cite{hillmich2021decision} proposed a decision diagrams method to generate an estimation and  Hadfield~\cite{hadfield2021adaptive} proposed an adaptive Pauli Shadow algorithm to generate an estimation.
While  similar problems are considered, the techniques are different and could be compared. 
{After our work, Shlosberg et al.~\cite{shlosberg2021adaptive} and Yen et al.~\cite{yen2022deterministic} apply the shadow and overlapped grouping ideas to commute measurement respectively, and their results have better performances with the cost of a polynomial-size quantum circuit.
}

%\printbibliography
\bibliographystyle{plainnat}
\bibliography{ref}

\appendix

\section{\label{app:UpperAlgs}Sample complexity}

In this section, we theoretically analyze the upper bound of the number of copies for the existing algorithms, $l_1$-sampling, grouping, and classical shadow algorithm theoretically. We further provide the variance of the derandomization algorithm and show the relation to our overlapped grouping measurement method.

Let $\hat{v}_1, \ldots, \hat{v}_T$ be the estimations after $T$ independent samples. Let $\hat{v} = \pbra{\hat{v}_1 + \cdots + \hat{v}_T}/T$ be the expectation of these samples.
Then by Chebyshev inequality, we have
\begin{align}
    \Pr\sbra{\abs{\hat{v} - \tr(\rho Q)} \geq \varepsilon }\leq \frac{\text{Var}(\hat{v})}{\varepsilon^2} = \frac{\text{Var}\pbra{\hat{v}_1}}{T\varepsilon^2}.
\end{align}
Hence the error can be bounded to $\varepsilon$ with probability $\delta$ when the number of samples is Var$\pbra{\hat{v}_1}/\pbra{\delta\varepsilon^2}$. By the definition 
 of $\hat{v}$ in Eq. \eqref{eq:UnifiedRep},
the variance of $\hat{v}$ equals $\Ebb\sbra{\hat{v}^2}-\tr(\rho \Omat)^2\leq \Ebb\sbra{\hat{v}^2}$. 

\emph{$l_1$-sampling.---}
The variance of $\hat{v}_1$ generated by $l_1$-sampling can be bounded to
\be
\text{Var}\pbra{\hat{v}_1}\leq \sum_{j=1}^m \alpha_j^2 \frac{\vabs{\alpha}_1}{\abs{\alpha_j}} =\vabs{\alpha}_1^2 . 
\ee
Hence $T \geq \frac{\vabs{\alpha}_1^2}{\delta\varepsilon^2}$ suffices to give an estimation with error less than $\varepsilon$  and success probability $1-\delta$.

\emph{Grouping.---}
The variance of $\hat{v}_1$ generated by grouping method satisfies
\begin{align}
\begin{aligned}
\text{Var}\pbra{\hat{v}_1} &\leq \sum_{j,k}\alpha_j \alpha_k f_{\text{Group}}\pbra{Q^{(j)}, Q^{(k)}, \Kcal}\tr\pbra{\rho Q^{(j)}Q^{(k)}}\\
&\leq \pbra{ \sum_{j}\alpha_j f_{\text{Group}}(Q^{(j)}, Q^{(j)}, \Kcal)}^2 \\
&= \pbra{\sum_{k = 1}^s \frac{1}{\Kcal_j} \sum_{j:Q^{(j)}\in \bm e_k}\alpha_j}^2 \\
&= \vabs{\alpha}_1^2\pbra{\sum_{k = 1}^s\frac{1}{\vabs{\alpha_{\bm e_k}}_1}\sum_{j:Q^{(j)}\in \bm e_k} \alpha_j}^2,
\end{aligned}
\end{align}
where $\vabs{\alpha_{\bm e_k}}_1 =\sum_{j:Q^{(j)}\in \bm e_k}\abs{\alpha_{j}}.$

\emph{Classical shadow.---}
For classical shadow algorithm, the variance of the generated estimation $\hat{v}_1$ satisfies
\begin{align}
\begin{aligned}
\text{Var}\pbra{\hat{v}_1} &\leq \sum_{j,k}\alpha_j \alpha_k f_{\text{LBCS}}\pbra{Q^{(j)}, Q^{(k)}, \Kcal}\tr\pbra{\rho Q^{(j)}Q^{(k)}} \\
&\leq \pbra{ \sum_{j}\alpha_j \prod_{k = 1}^{n}\frac{1}{\chi_k\pbra{Q_k^{(j)}}}}^2 \\
&\leq 3^{\text{local(\Omat)}}\pbra{\sum_{j=1}^m \alpha_j}^2 
\end{aligned}
\end{align}
where local$(\Omat) = \max_{j=1}^m \text{ local}(Q^{(j)})$, and the locality of $Q^{(j)}$: local$\pbra{Q^{(j)}} = \# \cbra{k |Q^{(j)}_k \ne \Ibb}$ is the number of qubits $k$ such that $Q^{(j)}_k$ is not identity.
Therefore, $T \geq \frac{3^{\text{local}(\Omat)}\pbra{\sum_j\alpha_j}^2}{\delta\varepsilon^2}$ suffices to give an estimation with  error $\varepsilon$ and success probability $1-\delta$.

\section{LDF Grouping method\label{app:LDFGroup}}

In the LDF Grouping method, we mapped the observables and the ``compatible with'' relationship to a graph. In specific, we denote an observable as a vertex, and if two observables do not have any ``compatible with'' relationship, we connect them with an edge. Then we can obtain a graph $G(V,E)$, where the number of vertices is equal to $m$. Next, we can proceed the grouping method as follows.

\begin{itemize}
\item [(1)] Sorting vertices $V_{1},V_{2},\ldots, V_{m}$  in the descending order of its degree.
\item [(2)] Repeat the following step until all of the vertices are in one of the sets.
\item [(3)] For $j$ goes from $1$ to $m$, if $V_{j}$ is not in any set, then add $V(j)$ to a set such that there is no edge between $V(j)$ and any other vertices in this set. If such a set does not exist, then add $V(j)$ to a new set.
\end{itemize}
After the above process and changing $V_{i}$ into $Q^{(i)}$, we can generate the grouping sets which satisfy any two observables are compatible with each other.

\section{Greedy overlapped grouping strategy\label{app:OGMGroup}}

Aside from the overlapped set generation strategy in Algorithm 1 of the main text, we proposed an alternative strategy that is slightly different in selecting the observables for a set, denoted as Grouping version 2. The main difference is in ``the sequence of observables'' adding to a new set.  The new strategy has a potentially better performance but it needs more time since it has a larger number of sets.
We further add a token to each observable to represent how many times we can visit this observable to avoid the explosion of the number of sets. Note that if there is no restriction, the observables in the tail of the sequence will be much more difficult to be added to existing sets, hence the number of sets could be very large.

Let the token of an observable $Q^{(k)}$ be $U_k = 2^{d-1}$, where $d$ is the number of digits of $\floor{\abs{\alpha_k}/\text{Min}{\text{Weight}}}$, and $\text{Min}{\text{Weight}} = \min_{j\leq m} \abs{\alpha_j}$ is the minimum weight.
This new version of grouping strategy is depicted in Algorithm~\ref{alg:OverlapSetV2}.

\SetNlSty{textbf}{}{\quad}
\IncMargin{1em}
\begin{algorithm}
\SetKw{KwTo}{to}
\SetKwFor{While}{while}{:}{}
\SetKwFor{For}{for}{:}{}
\SetKwIF{If}{ElseIf}{Else}{if}{:}{elif}{else:}{}
\SetKwInOut{Input}{Input}
\SetKwInOut{Output}{Output}
\Input{$n,m$ and $Q^{(1)},\ldots, Q^{(m)}, \alpha_1,\ldots, \alpha_m$, $U_1,\ldots, U_m$.}
\Output{$\cbra{P^{(s)}}$ with initial probabilities $\cbra{\Kcal_s}$.}
\emph{Sorting all of the observables $\{Q^{(j)}\}$ to the descending order according to their weights $\abs{\alpha_j}$}\;
$j\leftarrow 1$ and $s\leftarrow 1$\;
\While{$\exists Q^{(j)}$ which is not in any sets}{
\qquad\emph{Let $Q^{(j)}$ be the first observable in the sorted sequence which has never appeared in any sets and add it into a new set $\bm e_s$}\;
\qquad\emph{$s\leftarrow s+1$}\;
\qquad\emph{Initialize the measurement of $\bm e_s$ as $P^{(s)}\leftarrow Q^{(j)}$}\;
\qquad\For{$k\leftarrow 1$ \KwTo $m$ }{

\qquad\qquad\If{$Q^{(k)}$ is compatible with $P^{(s)}$ and $k\ne j$, the number of sets $Q^{(k)}$ appeared is less than token $U_k$}{
\qquad\qquad\qquad\emph{Add $Q^{(k)}$ into set $\bm e_s$, and update $P^{(s)}$ to $ P^{(s)} \bigvee Q^{(k)}$}
\tcc*{$P = Q\bigvee R$ is defined as  $P_i = Q_i$ if $Q_i = R_i$ and $P_i = Q_iR_i$ otherwise.}
    }
}
\qquad\emph{Let the initial probability of $P^{(s)}$ be the summation of the weight of all observables in this set}\;
\qquad\For{$k\leftarrow 1$ \KwTo $j-1$}{
\qquad\qquad\If{$Q^{(k)}$ is compatible with $P^{(s)}$ and not in $\bm e_s$}{
\qquad\qquad\qquad\emph{Add $Q^{(k)}$ into set $\bm e_s$, and update $P^{(s)}$ to $ P^{(s)} \bigvee Q^{(k)}$}\;
}
}
}
\caption{Second version of overlapped set generation.
}
\label{alg:OverlapSetV2}
\end{algorithm}
\DecMargin{1em}

 We compared the error of the estimation generated by grouping versions 1 and 2 in Table \ref{tab:ComGroupVersion}. The table shows that when we take more consideration for the observables with larger weight, we have better optimized results, while this would give us a longer optimization time because of the expansion of sets (optimized parameters).

\begin{table}[htbp]
    \centering
        \caption{Comparison for the errors of the OGM algorithms with grouping version 1 and 2. The error is estimated with 1000 samples for the ground state of the corresponding molecules under JW-encoding.}
    \begin{tabular}{c|c c c c c c}
    \hline
       molecule  &  H$_2$ (4) & H$_2$ (8) & LiH & BeH$_2$ & H$_2$O & NH$_3$\\
       \hline\hline
       error (V1)  & 0.011 & 0.051 & 0.036 & 0.072 & 0.129 & 0.151\\
       \hline
       error (V2) & 0.013 & 0.047 & 0.021 & 0.051 & 0.121 & 0.115
       \\
       \hline
    \end{tabular}
    \label{tab:ComGroupVersion}
\end{table}

\section{Variance for the partially derandomized strategy}
\label{appendix:derand}
Suppose that we have determined the measurement basis set $\Mcal$, in which the number of measurement basis  $P^{(k)}$ is assigned as $M_k$, and the total number of measurements is $T = \sum_k M_k$. We show the variance for the partially derandomized strategy with given measurements $T$.
We denote $\bm e_k$ containing all $Q^{(k)}$ element-wise commute with basis $P^{(k)}$, and denote $s_k$ as the total number of times that $Q^{(k)}$ is effectively measured, which is given by
$s_k = \sum_{P\in\Mcal} \delta_{Q^{(k)}\triangleright P}$.  Let ${t}_{j,k}$ be the measurement outcome of the $j$th observable $Q^{(j)}$ measured with the basis $P^{(k)}$ is $+1$. The measurement outcome associated with the measurement $P^{(k)}$ for observable  $Q^{(j)}$ is thus $\hat{v}_{j,k} = 2{t}_{j,k}/M_k - 1$. As such, the estimator can be expressed by
\begin{equation}
    \hat v = \sum_k \hat{v}_k =   \sum_k \sum_{j:Q^{(j)} \in \bm{e}_k }  \frac{\alpha_j M_k}{s_j}   \hat{v}_{j,k}.
    \label{eq:estimate_derand}
\end{equation}
One can check that if $M_k>0$ ($\forall k$) holds, i.e., every observable is assigned at least one measurement basis (one sample), the estimation in Eq.~\eqref{eq:estimate_derand}
is unbiased.

The variance of the estimator $\hat v$ is given by
\begin{equation}
\textrm{Var}[\hat v] = \sum_{k} M_{k} \sum_{j,j':Q^{(j)},Q^{(j')} \in \mathcal
S_{k}} \frac{1}{s_{j} s_{j'}} \alpha_{j} \alpha_{j'} \operatorname{Cov}\left(\hat{v}_{j,k}, \hat{v}_{j',k}\right).
\label{eq:derand}
\end{equation}
Here, we use the fact that measurement outcomes obtained from different  $P^{(k)}$ are independent  since the measurements $P^{(k)}$ are independent of each other.
We also note that the outcomes $\hat{v}_{j,k}$ are correlated, so the variance in Eq.~\eqref{eq:derand} depends on the covariance  $\operatorname{Cov}\left(\hat{v}_{j,k}, \hat{v}_{j',k}\right)$.

\section{Relationship between cost function and variance}
\label{app:CFandVariance}
Let $o_j := f(P,Q^{(j)}, \Kcal)\mu\pbra{P,Q^{(j)}}$ be the estimation of $\tr\pbra{\rho Q^{(j)}}$.
The cost function in the manuscript is selected as
$\sum_j{\alpha_j}^2\Ebb\sbra{o_j^2} = \sum_{j} {\alpha_j}^2/\chi(Q^{(j)})$, to evaluate $\Var{v}$.
In the following, we prove that $\Var{v}\leq m\sum_j{\alpha_j}^2\Ebb\sbra{o_j^2}$.

By Eq. \eqref{eq:VarOGM}, we have
\begin{align}
    \mathrm{Var}(\hat{v})&\leq\sum_{j,k} \alpha_j\alpha_k g(Q^{(j)},Q^{(k)})\tr\pbra{\rho Q^{(j)}Q^{(k)}}
    \label{eq:var_rho}
    \\
    &= \sum_j \alpha_j^2 \Ebb\sbra{o_j^2} + \sum_{j\ne k} \alpha_j \alpha_k \frac{\sum_{P:Q^{(j)}\triangleright P, Q^{(k)}\triangleright P}\Kcal(P)}{\chi(Q^{(j)})\chi(Q^{(k)})}\tr\pbra{\rho Q^{(j)}Q^{(k)}}
\end{align}
and 
\begin{align}
    &\sum_{j\ne k} \alpha_j \alpha_k \frac{\sum_{P:Q^{(j)}\triangleright P, Q^{(k)}\triangleright P}\Kcal(P)}{\chi(Q^{(j)})\chi(Q^{(k)})}\tr\pbra{\rho Q^{(j)}Q^{(k)}}\\
    &\leq \sum_{j\ne k} \alpha_j \alpha_k \frac{\sum_{P:Q^{(j)}\triangleright P, Q^{(k)}\triangleright P}\Kcal(P)}{\chi(Q^{(j)})\chi(Q^{(k)})}\\
    &\leq \sum_{j\ne k} \alpha_j \alpha_k \frac{\min(\chi(Q^{(j)}),\chi(Q^{(k)}))}{\chi(Q^{(j)})\chi(Q^{(k)})}\\
    &=\sum_{j\ne k} \alpha_j \alpha_k \frac{1}{\max(\chi(Q^{(j)}),\chi(Q^{(k)}))}\\
    &\leq \sum_{j< k} (\alpha_j^2 + \alpha_k^2) \frac{1}{\max(\chi(Q^{(j)}),\chi(Q^{(k)}))}\\
     &\leq \pbra{m-1}\sum_j \alpha_j^2\frac{1}{\chi(Q^{(j)})}\\
     \label{eq:var_bounded}
\end{align}

As shown in Eq. \eqref{eq:var_rho}, $\Ebb\sbra{o_jo_k}$ is associated with the quantum state $\rho$, and it is unknown for us, while it can be bounded by Eq. \eqref{eq:var_bounded}. Hence we utilize $\sum_{j}\abs{\alpha_j}^2 \Ebb\sbra{o_j^2}$ as our cost function.

\section{Variance comparison}
\label{app:variance_comp}
We compare the variances of different measurement schemes in Table~\ref{tab:VarianceComp_main}, {where the initial point is directly chosen from Algorithm \ref{alg:OverlapSet}}. {Here we generate the measurement distribution by choosing $T = 1000$ in Eq. \eqref{eq:lfFinal}.}
With a negligible small error $\varepsilon_0$, we find that our method has a much smaller variance than that of  LDF-grouping, and LBCS. The improvement becomes more prominent for larger molecules (approximately one order compared to the classical shadow methods), which indicates its effectiveness for large practical problems {with a limited number of measurements.}

\begin{table}[htbp]
\centering
    \caption{{Estimation variances computed on the ground state of the molecular Hamiltonian (under the JW-encoding). The number of qubits is noted in the bracket in the first column, from 4 qubits to 14 qubits. 
    The first {3} columns are the variance of  LDF-grouping and OGM algorithms. 
    The last column is the initial error of OGM.}}
    \begin{tabular}{c| c c c| c}
    \hline
    \multirow{2}{*}{Molecule} & 
    & \multicolumn{2}{c|}{variance} & $\varepsilon_0$\\
    \cline{2-5}
   & LDF & LBCS~\cite{hadfield2020measurements} & OGM & OGM\\ 
       \hline\hline
        H$_2$ ($4$) &0.402 &1.86 & \textbf{0.424} & 0\\
        \hline
        H$_2$ ($8$) & 22.3 & 17.7 & \textbf{5.51} & 0\\
        \hline
        LiH ($12$) & 54.2 & 14.8& 
        \textbf{3.09} & $9.6\cdot 10^{-5}$\\
        \hline
        BeH$_2$ ($14$) & 135 & 67.6 & \textbf{15.44} & $6.6\cdot 10^{-4}$\\
   \hline
        H$_2$O ($14$) & 1040 & 258 & \textbf{39.64} & 0.053\\
        \hline
    \end{tabular}
    \label{tab:VarianceComp_main}
\end{table}

\section{Numerical results and discussions\label{app:numericalError}}
In this section, we numerically show the advantages of our OGM algorithm compared with $l_1$-sampling, LDF Grouping, LBCS, and the derandomized CS algorithm by computing the corresponding variances and errors.

Table \ref{tab:VarianceComp} shows the comparison of variance of OGM, $l_1$-sampling, LDF Grouping, and LBCS algorithms under different fermionic-to-qubit encodings, including JW, bk, and parity encodings. The last column is the deviation of the OGM algorithm after a small perturbation of the cost function as introduced in the main text.

\begin{table}[htbp]
\centering
    \caption{Variance of estimations computed on the ground state of the molecular Hamiltonian with different encoding methods. The first 4 columns are the variance of $l_1$-sampling, LDF-grouping, and OGM algorithms. The last column is the corresponding initial error of algorithm OGM.}
    \begin{tabular}{c| c c c c| c}
    \hline
    \multirow{2}{*}{Molecule} & 
    & \multicolumn{3}{c|}{variance} & $\varepsilon_0$\\
    \cline{2-6}
    & $l_1$ & LDF & LBCS~\cite{hadfield2020measurements} & OGM & OGM\\ 
       \hline\hline
        H$_2$($4$jw) &2.536 &0.402 &1.86 & \textbf{0.424} & 0\\
        H$_2$($4$bk) & 2.539 &0.193 &0.541 & \textbf{0.297} & 0\\
        H$_2$($4$parity)  &2.539 &0.193 &0.541 & \textbf{0.297} & 0\\
        \hline
        H$_2$ ($8$jw) & 119.8 & 22.3 & 17.7 & \textbf{5.51} & 0\\
        H$_2$ ($8$bk) & 124.9 & 38.4 &  19.5 & \textbf{5.66} & $8\cdot 10^{-6}$\\
        H$_2$ ($8$parity)& 124.9 & 38.0 & 18.9 & \textbf{6.96} & 0\\
        \hline
        LiH ($12$jw)& 145.4 & 54.2 & 14.8& 
        \textbf{3.09} & $9.6\cdot 10^{-5}$\\
        LiH ($12$bk)& 138.5 & 75.5 &  68.0 & \textbf{3.53} &$1.74\cdot 10^{-3}$\\
        LiH ($12$parity)& 138.5 & 85.8 & 26.5 & \textbf{5.52} & $3.03 \cdot 10^{-4}$\\
        \hline
        BeH$_2$ ($14$jw) & 453.4 & 135 & 67.6 & \textbf{15.44} & $6.6\cdot 10^{-4}$\\
        BeH$_2$ ($14$bk) & 464.9 & 197& 238 & \textbf{17.84} & $3.6\cdot 10^{-3}$\\
        BeH$_2$ ($14$parity) & 446 & 239 & 130 & \textbf{17.28} & $3.42\cdot 10^{-3}$\\
        \hline
        H$_2$O ($14$jw) & 4367 & 1040 & 258 & \textbf{39.64} & 0.053\\
        H$_2$O ($14$bk) & 5141 & 2090 & 1360 & \textbf{81.59}& 0.065\\
        H$_2$O ($14$parity) & 5017 & 2670 & 429 & \textbf{42.91} & 0.065\\
        \hline
    \end{tabular}
    \label{tab:VarianceComp}
\end{table}
  
We compare the estimation accuracy with $1000$ measurements for molecules H$_2$, LiH, BeH$_2$ and H$_2$O with JW, BK and parity encodings in Table \ref{tab:ErrorComp}, where the initialized probability for the OGM algorithm is chosen directly from Algorithm \ref{alg:OverlapSet}.

\begin{table}[htbp]
    \centering
        \caption{Error of estimations with 1000 samples for the ground state of the corresponding molecules with LDF-grouping, LBCS, Derandmized, and OGM algorithms.}
    \begin{tabular}{c| c c c c }
    \hline
       Molecule & LDF & LBCS~\cite{hadfield2020measurements} & Derand~\cite{huang2021efficient} & OGM\\
       \hline\hline
        H$_2$($4$jw)  & 0.019 &0.043 & 0.018 & \textbf{0.011}\\
        H$_2$($4$bk)  & 0.014&0.025 & 0.029 & \textbf{0.016}\\
        H$_2$($4$parity) & 0.016 & 0.016 & 0.027 &  \textbf{0.017}\\
        \hline
        H$_2$ ($8$jw) & 0.149 &0.128 & 0.067 &  \textbf{0.051}\\
        H$_2$ ($8$bk) & 0.203 &0.143 & 0.049 &  \textbf{0.041}\\
        H$_2$ ($8$parity) & 0.214 & 0.138 & 0.058 &  \textbf{0.053}\\
        \hline
        LiH ($12$jw)  & 0.231 & 0.122 & 0.063 &  \textbf{0.036}\\
        LiH ($12$bk)  & 0.293 & 0.283 & 0.067 &  \textbf{0.033}\\
        LiH ($12$parity)  & 0.294 & 0.145 & 0.061 & \textbf{ 0.039}\\
        \hline
        BeH$_2$ ($14$jw)  & 0.426 & 0.275 & 0.103 &  \textbf{0.072}\\
        BeH$_2$ ($14$bk)  & 0.462 & 0.527 & 0.104 & \textbf{ 0.082}\\
        BeH$_2$ ($14$parity) & 0.547 & 0.355 & 0.104 & \textbf{ 0.068}\\
        \hline
        H$_2$O ($14$jw)  & 1.090 & 0.549 & 0.257 & \textbf{0.129 }\\
        H$_2$O ($14$bk)  & 1.619 & 1.073 & 0.274 & \textbf{0.174 }\\
        H$_2$O ($14$parity)  & 1.476 & 0.708 & 0.336 &  \textbf{0.132}\\
        \hline
        NH$_3$ ($16$jw)  & 1.063 & 0.484 & 0.225 &  \textbf{0.151}\\
        NH$_3$ ($16$bk)  & 0.820 & 0.571 & 0.282 &  \textbf{0.103}\\
        NH$_3$ ($16$parity) & 2.018 & 0.663 & 0.294 &\textbf{ 0.173}\\
        \hline
    \end{tabular}
    \label{tab:ErrorComp}
\end{table}

\section{Error of the estimation\label{app:errorAnaly}}

We leverage root-mean squared error to quantify the error of the estimation. 
In each iteration, we generate an estimation $\hat{v}_i$ by independently performing $T$ measurements of initial state $\rho$.
For $N$ independent repetitions, we get the average error of $T$ samples as
\be
\varepsilon_{v} = \sqrt{\frac{1}{N}\sum_{i = 1}^N \pbra{\hat{v}_i - \tr(\rho O)}^2}.
\ee
We plot the figure to show that the average error will fluctuate in a small range after more than 10 iterations.

\begin{figure}[htbp]
    \centering
    \includegraphics[width = 0.56\textwidth, trim = 1cm 6cm 1cm 6cm, clip]{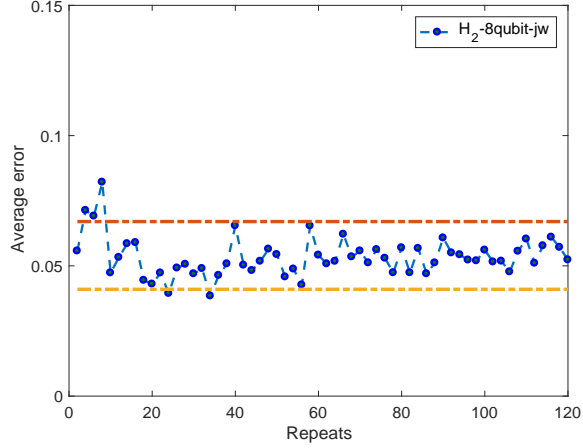}
    \caption{The average errors of different repetitions $N$ for H$_2$ molecule (8-qubit).}
    \label{fig:RepeatError}
\end{figure}
In our simulation experiment of the main file (Table \ref{tab:ErrorComp_main} and Fig. \ref{fig:err_molecules}), we choose $N = 100$ for the molecules in which the number of qubits is less than 14. For NH$_3$ molecule, we let $N = 20$. 
Nevertheless, the experimental result on the IBM Q device is an average of 12 rounds due to the high cost of the pending time.

\section{Experimental results}
\label{app:experimental_test}
 
{The numerical study ignored device errors, and how the noise in realistic hardware affects the measurement efficiency is critical for studying their practical performance with realistic quantum devices. 
To further demonstrate the advantage of our OGM method with current quantum devices, we implement and compare the measurement schemes on the IBM quantum cloud hardware with device imperfections.} We aim to estimate $\tr(\ket{\psi}\bra{\psi} \Omat_{H_2})$ with the GHZ state $\ket{\psi} = ({\ket{0000} + \ket{1111}})/{\sqrt{2}}$
 and the four-qubit Hamiltonian $\Omat_{H_2}$ of the $H_2$ molecule under the JW-encoding. We note that the GHZ state has a much larger variance compared to the ground state, and thus could be a suitable testbed to compare the performance of different measurement schemes. In Fig.~\ref{fig:CompIBM}, we compare estimation errors using the $l_1$-sampling, LDF-grouping, LBCS, derandomized CS, and the OGM method with a different number of copies (samples) of the prepared entangled state. We evaluate the error by comparing the reference results obtained using the OGM method with 49140 samples.

 \begin{figure}[t]
    \centering
    \includegraphics[width = 0.8\textwidth]{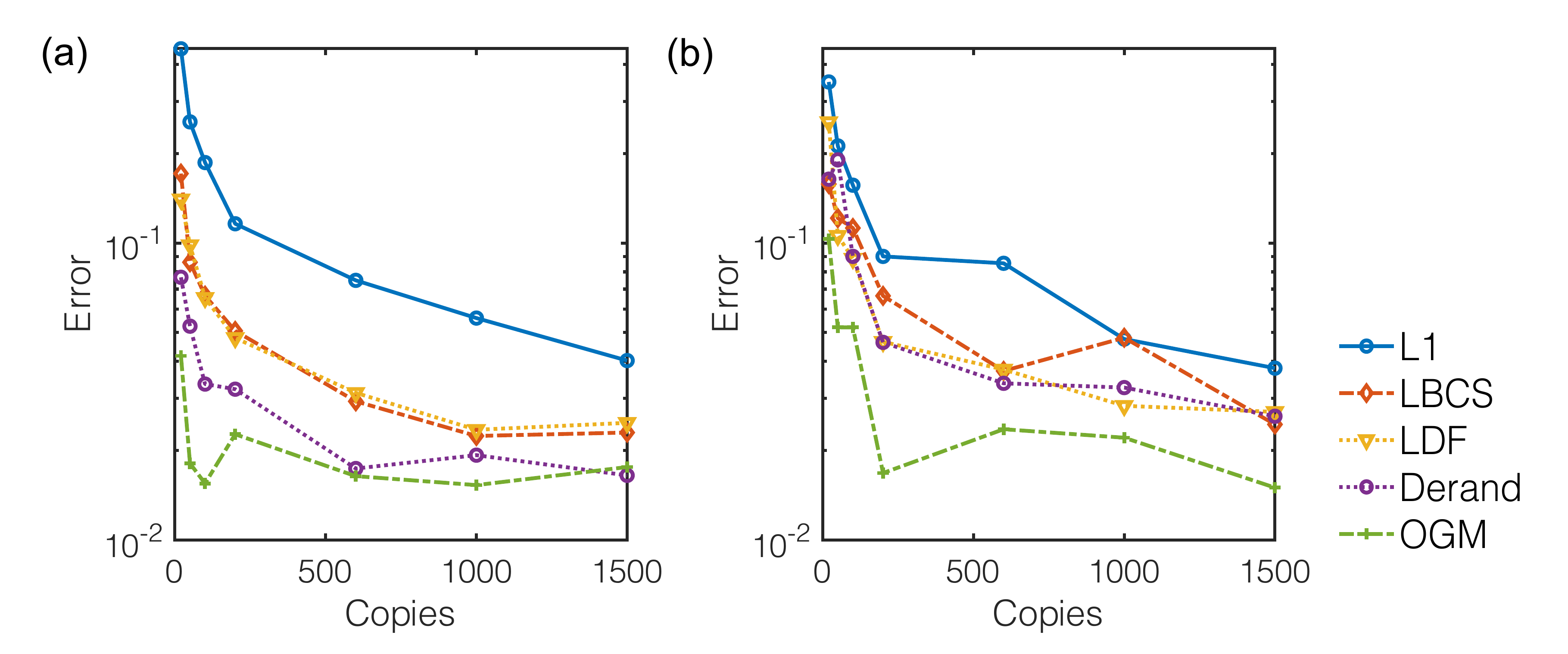}
    \caption{Comparison of OGM with $l_1$-sampling, LDF-grouping, LBCS, and derandomized CS using  (a) the IBM classical simulator (Repeat 100 rounds) and (b) the IBMQ cloud devices ibmq\_athens (Repeat 12 rounds).}
    \label{fig:CompIBM}
\end{figure}

\end{document}